\documentclass[letterpaper, 10 pt, conference]{ieeeconf} 

\IEEEoverridecommandlockouts               
\usepackage[dvipsnames]{xcolor}
\usepackage{cite}
\usepackage{amsmath,amssymb,amsfonts}
\usepackage{algorithmic}
\usepackage{dutchcal}
\usepackage{graphicx}
\usepackage[OT2,T1]{fontenc}
\usepackage{MnSymbol}
\usepackage{lipsum}
\usepackage{bbm}
\usepackage{color}
\newtheorem{theorem}{Theorem}

\newtheorem{definition}{Definition}

\newtheorem{proposition}{Proposition}

\newtheorem{remark}{Remark}
\DeclareSymbolFont{cyrletters}{OT2}{wncyr}{m}{n}
\DeclareMathSymbol{\Sha}{\mathalpha}{cyrletters}{"58}
\DeclareMathAlphabet{\mathscr}{OT1}{pzc}{m}{it}
\usepackage{textcomp}

\title{\LARGE \bf
Harmonic Pole Placement
}

\author{Pierre Riedinger and Jamal Daafouz
\thanks{This work is supported by HANDY project ANR-18-CE40-0010-02}
\thanks{P. Riedinger and J. Daafouz are with Universit\'e de Lorraine, CNRS-CRAN UMR 7039, 2, avenue de la for\^et de Haye, 54516 Vandoeuvre-l\`es-Nancy Cedex, France.
    {\tt\small Pierre.Riedinger@univ-lorraine.fr, Jamal.Daafouz@univ-lorraine.fr }}%
}

\begin{document}

\maketitle
\thispagestyle{empty}
\pagestyle{empty}

\begin{abstract}
In this paper, we propose a method to design state feedback harmonic control laws that assign the closed loop poles of a linear harmonic model to some desired locations. The procedure is based on the solution of an infinite-dimensional harmonic Sylvester equation under an invertibility constraint. We provide a sufficient condition to ensure this invertibility and show how this infinite-dimensional Sylvester equation can be solved up to an arbitrary small error. The results are illustrated on an unstable linear periodic system. We also provide a counter-example to illustrate the fact that, unlike the classical finite dimensional case, the solution of the Sylvester equation may not be invertible in the infinite dimensional case even if an observability condition is satisfied.
\end{abstract}

\section{Introduction}
 Harmonic modeling and control is a topic of theoretical and practical interest in many application domains such as energy management or embedded systems to mention few \cite{Farkas,Bolzern,Sanders,Montagnier,Sinha,Wereley_1990,Demiray,Almer,Zhou2011,Zhou2008,Mattavelli,Chavez,Almer2}. 
In a recent paper \cite{Blin}, a unified and coherent mathematical framework for harmonic modelling and control has been proposed. Basically, the harmonic modeling of a periodic system leads to an equivalent time invariant model of infinite dimension whose states (also called phasors) are the coefficients obtained by applying a sliding Fourier decomposition. One of the main results of \cite{Blin} established a strict equivalence between these two models. In this framework, the analysis and design are considerably simplified since all the methods established for time-invariant systems can be a priori applied. \\

Recently, results related to spectral properties of the harmonic state operator, explicit Floquet factorization and practical solutions of harmonic Lyapunov and Riccati equations have been established in \cite{Pierre2022}. Here, we exploit these results to address the problem of designing harmonic pole placement based control laws for Linear Time Periodic (LTP) systems. 
This question has been partially addressed in the past \cite{Kabamba} but without using an harmonic framework and the proposed approach was dedicated to Linear Time Invariant (LTI) systems controlled using a periodic control. To our knowledge, this is the first time that the harmonic pole placement problem is explicitely studied in its general formulation.

This paper is organized as follows. In Section II, we start with some prelimaries concerning harmonic modeling and the problem formulation. Section III is devoted to the main result which is a procedure for harmonic pole placement design. We show that the pole placement can be achieved by solving a harmonic Sylvester equation provided that the resulting solution is invertible. A sufficient condition is proposed to ensure this invertibility. We conclude this section by showing how, in practice, the Sylvester equation of the infinite dimensional problem can be solved up to an arbitrary error. This point is of importance since a naive approach consisting to apply a truncation on the infinite dimensional equation leads to erroneous numerical results as shown in \cite{Pierre2022}. A complete case study is given in section IV where an unstable LTP system is used to illustrate the results presented in the paper. Moreover, a counter-example is given to show that contrary to the finite dimensional case, the solution of the Sylvester equation may not be  invertible in the infinite dimensional case even if an observability condition is satisfied.

{\bf Notation: }The transpose of a matrix $A$ is denoted $A'$ and $A^*$ denotes its complex conjugate transpose $A^*=\bar A'$. The $n$-dimensional identity matrix is denoted $Id_n$. The infinite identity matrix is denoted $\mathcal{I}$. $A \otimes B$ is the Kronecker product of two matrices $A$ and $B$. $L^{p}$ (resp. $\ell^p$) denotes the Lebesgues spaces of $p-$integrable functions (resp. $p-$summable sequences) for $1\leq p\leq\infty$. $L_{loc}^{p}$ is the set of locally $p-$integrable functions i.e. on any compact set. The notation $f(t)=g(t)\ a.e.$ means almost everywhere in $t$ or for almost every $t$. We denote by $col(X)$ the vectorization of a matrix $X$, formed by stacking the columns of $X$ into a single column vector. Finally, $<\cdot,\cdot>$ refers to the scalar product in $\ell^2$.
\section{Preliminaries and problem formulation}


\addtolength{\textheight}{-3cm}  

We first start be recalling the definition of the sliding Fourier decomposition over a 
window of length $T$ and the so-called "Coincidence Condition" introduced in \cite{Blin}. 

\begin{definition}The sliding Fourier decomposition over a window of length $T$ from $ L^{2}_{loc}(\mathbb{R},\mathbb{C}^n)$ to $L^{\infty}_{loc}(\mathbb{R},\ell^2(\mathbb{C}^n))$ is defined by:
	$$X:=\mathcal{F}(x)$$
	where the time-varying infinite sequence $X$ is defined by:
	\begin{equation}
		t\mapsto X(t):=(\mathcal{F}(x_1)(t), \cdots,\mathcal{F}(x_n)(t))\label{hs}\end{equation}
	and where for $i:=1,\cdots,n$, the vector $\mathcal{F}(x_i):=(\cdots, X_{i,-1}, X_{i,0}, X_{i,1},\cdots)$, has infinite components $X_{i,k} $, $k\in \mathbb{Z}$ satisfying:
	$X_{i,k}(t):=\frac{1}{T}\int_{t-T}^t x_i(\tau)e^{-j\omega k \tau}d\tau.
	$
	The vector $X_k:=(X_{1,k}, \cdots, X_{n,k})$ is called the $k-$th phasor of $X$. \end{definition}
In the sequel, to distinguish a matrix function $P(\cdot)$ and its sliding Fourier decomposition $\mathcal{F}(P)$, we use the notation ${\bf P}:=\mathcal{F}(P)$ and $P(t)$ instead of $P(\cdot)$.

\begin{definition}\label{H} We say that $X$ belongs to $H$ if $X$ is an absolutely continuous function (i.e $X\in C^a(\mathbb{R},\ell^2(\mathbb{C}^n))$ and fulfills for any $k$ the following condition: \begin{equation}\dot X_k(t)=\dot X_0(t)e^{-j\omega k t} \ a.e.\label{dphasor2}\end{equation}
\end{definition}
Similarly to the Riesz-Fisher theorem which establishes a one-to-one correspondence between the spaces $L^2$ and $\ell^2$, the following "Coincidence Condition" establishes a one-to-one correspondence between the spaces $L_{loc}^2$ and $H$.
\begin{theorem}[Coincidence Condition \cite{Blin}]\label{coincidence}For a given $X\in L_{loc}^{\infty}(\mathbb{R},\ell^2(\mathbb{C}^n))$, there exists a representative $x\in L^2_{loc}(\mathbb{R},\mathbb{C}^n)$ of $X$, i.e. $X=\mathcal{F}(x)$, if and only if $X$ belongs to $H$.
\end{theorem}

\begin{definition}\label{bt}
	The block Toeplitz transformation of a $T-$periodic $n\times n$ matrix function $A\in L^{2}([0 \ T])$, denoted $\mathcal{A}:=\mathcal{T}(A)$, defines a constant $n\times n$ block Toeplitz and infinite dimensional matrix as follows: 
	$$\mathcal{A}:=\left(\begin{array}{cccc}
		\mathcal{A}_{11} & \mathcal{A}_{12} & \cdots & \mathcal{A}_{1n} \\
		\mathcal{A}_{21} & \mathcal{A}_{22} & & \vdots \\
		\vdots & & \ddots& \vdots \\
		\mathcal{A}_{n1} & \cdots & \mathcal{A}_{n(n-1)} & \mathcal{A}_{nn}\end{array}\right)$$ where the infinite matrices $\mathcal{A}_{ij}:=\mathcal{T}(a_{ij})$, $i,j:=1,\cdots,n$, are the Toeplitz transformation of the $(i,j)$ entry $a_{ij}(t)$ of the matrix $A(t)$:
	\begin{align*}
		\mathcal{T}(a_{ij}):=
		\left[
		\begin{array}{ccccc}
			\ddots & & \vdots & &\udots \\ & a_{ij,0} & a_{ij,-1} & a_{ij,-2} & \\
			\cdots & a_{ij,1} & a_{ij,0} & a_{ij,-1} & \cdots \\
			& a_{ij,2} & a_{ij,1} & a_{ij,0} & \\
			\udots & & \vdots & & \ddots\end{array}\right],\end{align*}
	with $a_{ij,k} :=\frac{1}{T}\int_{t-T}^t a_{ij}(\tau)e^{-j\omega k \tau}d\tau$. \\
\end{definition}
 
We need to recall some key results from \cite{Blin}. Under the "Coincidene Condition" of Theorem~\ref{coincidence}, it is established in \cite{Blin} that any periodic system having solutions in Carath\'eodory sense can be transformed by a sliding Fourier decomposition into a time invariant system. For instance, consider $T-$periodic functions $A(\cdot)$ and $B(\cdot)$ respectively of class $L^2([0\ T],\mathbb{C}^{n\times n})$ and $L^{\infty}([0\ T],\mathbb{C}^{n\times m})$ and let the linear time periodic system:
\begin{align}\dot x(t)=A(t)x(t)+B(t)u(t)\quad x(0):=x_0\label{ltp}\end{align}
If, $x$ is a solution associated to the control $u\in L_{loc}^2(\mathbb{R},{\mathbb{C}^m)}$ of the linear time periodic system (\ref{ltp}) then, $X:=\mathcal{F}(x)$ is a solution of the linear time invariant system:
\begin{align}
	\dot X(t)=(\mathcal{A}-\mathcal{N})X(t)+\mathcal{B}U(t), \quad X(0)=\mathcal{F}(x)(0) \label{ltih}
\end{align}
where $\mathcal{A}:=\mathcal{T}(A)$, $\mathcal{B}:=\mathcal{T}(B)$ and 
\begin{equation}\mathcal{N}:=Id_n\otimes diag(j\omega k,\ k\in \mathbb{Z})\label{N}\end{equation}
Reciprocally, if $X\in H$ is a solution of \eqref{ltih} with $U\in H$, then its representative $x$ 
(i.e. $X=\mathcal{F}(x))$ is a solution of \eqref{ltp}. Moreover, for any $k\in\mathbb{Z}$, the phasors $X_k \in C^1(\mathbb{R},\mathbb{C}^n)$ and $\dot X\in C^0(\mathbb{R},\ell^{\infty}(\mathbb{C}^n))$. As the solution $x$ is unique for the initial condition $x_0$, $X$ is also unique for the initial condition $X(0):=\mathcal{F}(x)(0)$. In addition, it is proved in \cite{Blin} that one can reconstuct time trajectories from harmonic ones, that is:
\begin{align}\label{recos} x(t)&=\mathcal{F}^{-1}(X)(t):=\sum_{p=-\infty}^{+\infty} X_p(t)e^{j\omega p t}+\frac{T}{2}\dot X_0(t)\end{align}
where $X_{k}=(X_{1,k}, \cdots, X_{n,k})$ for any $k\in \mathbb{Z}$.\\

We recall also the following results where the proofs can be found in \cite{Gohberg} (Part V p.p. 562-574).
\begin{theorem}\label{borne} Let $A(t)\in L^2([0 \ T],\mathbb{C}^{n\times m})$. Then, $\mathcal{A}:=\mathcal{T}(A)$ is a bounded operator on $\ell^2$ if and only if $A\in L^{\infty}([0\ T],\mathbb{C}^{n\times m} )$.
	Moreover, the operator norm induced by the $\ell^2$-norm satisfies: $\|A(z)\|_{\ell^2}=\|\mathcal{A}\|_{\ell^2}=\|A\|_{L^{\infty}}$.
\end{theorem}	
\begin{theorem}\label{inverse}
	Let $A(t)\in L^\infty([0 \ T],\mathbb{C}^{n\times n})$. $\mathcal{A}$ is invertible if and only if there exists $\gamma>0$ such that the set $\{t: |\det(A(t)|<\gamma\}$ has measure zero. The inverse $\mathcal{A}^{-1}$ is determined by $\mathcal{T}(A^{-1})$.
	In addition, $\mathcal{A}$ is invertible if and only if $\mathcal{A}$ is a Fredholm operator,
	or equivalently in this setting if and only if
	there exists $c > 0$ such that
	$$\|\mathcal{A}x\|_{\ell^2} > c \|x\|_{\ell^2},\text{ for any } x \in \ell^2.$$
\end{theorem}
~~\\
The problem we want to solve in this paper can be formulated as follows. Consider the harmonic model (\ref{ltih}), design an harmonic state feedback control law of the form $U=-\mathcal{K}X$  that assigns the poles of the closed loop harmonic model to some desired locations and provide the corresponding representative in the time domain $u(t)=-K(t)x(t)$.

\section{Main result}
 In the sequel, we will use the following notion of controllability. 
\begin{definition}(see \cite{Rodman,Rabah}) System \eqref{ltih} is said to be exactly controllable at time $t$, if for all $X_0$, $X_1\in \ell^2$, there exists a control $U(t)\in L^2([0,t],\ell^2)$ such that the corresponding solution satisfies $X(t) = X_1$. 
\end{definition}
In this setting, exact controllability can be characterized as follows. 
\begin{proposition}
System \eqref{ltih} is exactly controllable at time $t$ if and only if
$$\exists \delta>0:\forall x\in \ell^2,\int_0^t\|\mathcal{B}^*e^{(\mathcal{A}-\mathcal{N})^*\tau}x\|^2 d\tau \geq \delta \| x \|^2$$
\end{proposition}
\begin{proof}see \cite{Rabah}
\end{proof}

\subsection{Harmonic Pole Placement}

%
%
The next Theorem provides a way to design a state feedback harmonic control law that assigns the poles of the closed loop system to some specified location.
%
%

\begin{theorem}\label{sylPP} Assume that the pair $(\mathcal{A}-\mathcal{N},\mathcal{B})$ is exactly controllable. Let $G(t)$ be a $T-$periodic function in $L^\infty$ and define $\mathcal{G}:=\mathcal{T}(G)$.
	Consider a $n-$dimensional Jordan normal form $\Lambda$ such that 
		the spectrum of $\mathcal{A}-\mathcal{N}$ and $\Lambda \otimes \mathcal{I}-\mathcal{N}$ have no common value i.e. $\sigma(\mathcal{A}-\mathcal{N})\cap \sigma(\Lambda \otimes \mathcal{I}-\mathcal{N})=\emptyset$. 
	
	Then, the bounded operator (on $\ell^2$) $\mathcal{P}$ is the unique solution of the harmonic Sylverter equation:
	\begin{align}(\mathcal{A}-\mathcal{N})\mathcal{P}-\mathcal{P}(\Lambda \otimes \mathcal{I}-\mathcal{N})&=\mathcal{B}\mathcal{G}\label{syl1}
	\end{align}
	if and only if $P:=\mathcal{T}^{-1}(\mathcal{P})$
	is the $T-$periodic solution (in Carath\'eodory sense) of the differential Sylvester equation
	\begin{align}
		\dot P(t)&=A(t)P(t)-P(t)\Lambda-B(t)G(t)\label{syl2}
	\end{align}

	If $\mathcal{P}$ is invertible, then $\mathcal{P}^{-1}$ is also a bounded operator on $\ell^2$ and ${P}^{-1}$ is a $T-$periodic solution in Carath\'eodory sense of the differential Sylvester equation:
	$$\dot P^{-1}(t)=-P^{-1}(t)(A(t)-B(t)K(t))+\Lambda P^{-1}(t)\ a.e.$$
	where $K(t):=G(t)P^{-1}(t) \in L^\infty$.
	Moreover, the harmonic state feedback $U:=-\mathcal{K}X$ with $\mathcal{K}:=\mathcal{G}\mathcal{P}^{-1}$, is a bounded operator on $\ell^2$ and assigns the poles of the closed loop harmonic system, that is:
	$$\mathcal{P}^{-1}(\mathcal{A}-\mathcal{N}-\mathcal{B}\mathcal{K})\mathcal{P}=(\Lambda \otimes \mathcal{I}-\mathcal{N}).$$
	In addition, taking $z(t):=P^{-1}(t)x(t)$ transforms the closed loop LTP system \eqref{ltp} with $u(t):=-K(t)x(t)$ and $K(t):=G(t)P^{-1}(t)\in L^\infty$
 into the LTI system
	\begin{equation}
		\dot z=\Lambda z\ a.e.\label{lti}
		\end{equation}

\end{theorem}
\begin{proof}
	If $\mathcal{P}$ solves \eqref{syl1} then $\mathcal{P}$ belongs trivially to $H$, and $P:=\mathcal{T}^{-1}( \mathcal{P})$
	satisfies the differential Sylvester equation \eqref{syl2} and reciprocally. This equivalence is obtained using similar steps to the proof of Theorem 5 in \cite{Blin}. 
	As the product $BG$ belongs to $L^\infty$, we also prove that $\mathcal{P}$ is a bounded operator on $\ell^2$ (see the proof of Theorem 5 in \cite{Blin}). 
	
	If $\mathcal{P}$ is invertible, then is ${P}$. Thus, $P$ and $P^{-1}:=\mathcal{T}^{-1}( \mathcal{P}^{-1})$ are both in $L^\infty$ (see Theorem~\ref{borne}). 
	As $G\in L^{\infty}$, the product $GP^{-1}$ belongs to $L^{\infty}$ and $\mathcal{K}:=\mathcal{G}\mathcal{P}^{-1}=\mathcal{T}(G)\mathcal{T}(P^{-1})=\mathcal{T}(GP^{-1})$ is a (block) Toeplitz and bounded operator on $\ell^2$.
	
	From \eqref{syl1}, we have :
	\begin{align}\mathcal{P}^{-1}(\mathcal{A}-\mathcal{B}\mathcal{K}-\mathcal{N})-(\Lambda \otimes \mathcal{I}-\mathcal{N})\mathcal{P}^{-1}&=0\label{syl3}
	\end{align}
	which implies (using similar steps of the proof of Theorem 5 in \cite{Blin}) that $P^{-1}:=\mathcal{T}^{-1}(\mathcal{P}^{-1})$ satisfies:
	$$\dot P^{-1}(t)=-P^{-1}(t)(A(t)-B(t)K(t))+\Lambda P^{-1}(t)\ a.e.$$
 in Carath\'eodory sense	with $K(t):=G(t)P^{-1}(t) \in L^\infty$.
		
	Furthermore, the control $U:=-\mathcal{K}X$ allows to assign the poles of the closed loop harmonic system to the set $\sigma:=\{\lambda+j\omega k:k\in \mathbb{Z},\lambda \in diag(\Lambda)\}$ according to the following relation:
	$\mathcal{P}^{-1}(\mathcal{A}-\mathcal{N}-\mathcal{B}\mathcal{K})\mathcal{P}=(\Lambda \otimes \mathcal{I}-\mathcal{N})$.
	
	Finally, taking $z(t):=P^{-1}(t)x(t)$, 
we have:
    \begin{align*}
	\dot z&=\dot P^{-1}(t)x+P^{-1}(t)\dot x \ a.e.\\
	&=\Lambda z(t) \ a.e.
	\end{align*}	
\end{proof}
\begin{remark}Theorem~\ref{sylPP} deserves some explanations concerning the choice of pole locations associated to $\Lambda$.
Such a limited choice in an infinite space is directly related to the spectral properties of the harmonic state operator $\mathcal{A}-\mathcal{N}$.
Indeed, as it has been shown in \cite{Pierre2022}, the spectrum of this operator is a discrete and unbounded set that contains only eigenvalues. This set is given by:
$$\sigma:=\{\lambda_p+j\omega k: k\in \mathbb{Z}, p:=1,\cdots,n\}$$ where the values $\lambda_p$, $p:=1,\cdots,n$ can be determined by a Floquet factorization of $A(t)$ provided in \cite{Pierre2022}. As a consequence, there are only $n$ locations to be fixed for pole placement. The above Theorem shows that this choice imposes the dynamics of $z:=P(t)^{-1}x(t)$.
\end{remark}
\begin{remark} 
	To explain the effect of pole placement for LTP systems, consider a one dimensional system and assume that, after a pole placement, equation~\eqref{lti} is given by $\dot z=-\alpha z$
	and $\dot x= (a(t)-b(t)k(t))x(t)$ where $x(t):=P(t)z(t)$.
	As $P(t)$ is periodic, we can write for $k:=1,2,\cdots$ :
	$$x(t_0+kT)=P(t_0)e^{-\alpha T k}z(t_0)$$ for any $t_0\in[0\ T]$. We see that the sequence $x(t_0+kT)$, $k:=1,2,\cdots$ for any $t_0\in[0\ T]$ has a decay rate $\gamma:=e^{-\alpha T}$. 
\end{remark}
\begin{remark}\label{inv}
Unlike the classical finite-dimensional pole placement problem where, in order to impose the invertibility of the solution of the Sylvester equation, an observability property is required \cite{Varga,Zbigniew}, such a necessary condition is not a sufficient in the infinite dimension case. We illustrate this using a counter example at the end of this paper.
\end{remark}

From a practical point of view, to check the invertibility of $\mathcal{P}$, we invoke Theorem~\ref{inverse} that is to check that 
 there exists $c>0$ such that $$\|\mathcal{P}x\|_{\ell^2}\geq c\|x\|_{\ell^2},\ x\in {\ell^2} $$
or equivalently, as $P(t)$ is continuous, that $|\det P(t)|\neq 0$ for any $t\in [0\ T]$. We also propose a sufficient condition to ensure the invertibility of $\mathcal{P}$ for a particular choice of poles locations. To this aim, let us consider $(V(t),J)$ the Floquet factorization \cite{Farkas,Zhou} of $A(t)$ so that $V^{-1}(t)A(t)V(t)=J$ where $J$ is a Jordan normal form and where $V(t)$ is $T-$periodic, invertible and absolutely continuous matrix function. Note that such a Floquet factorization always exists and can be explicitly computed thanks to the closed form formula given by Theorem~6 in \cite{Pierre2022}.
 
\begin{theorem}\label{suf}Consider a Floquet factorization given by $(V(t),J)$ and let $\mathcal{V}:=\mathcal{T}(V)$. The solution $\mathcal{P}$ of the Sylvester equation \eqref{syl1} is invertible
if $\mathcal{G}$ is set to $\mathcal{G}:=\mathcal{B}^*\mathcal{V}^{*-1}$ in \eqref{syl1} and
the poles locations correspond to $\Lambda:=-J^*-\alpha Id_n$ with $\alpha$ a real number.
\end{theorem}
\begin{proof}
To ease the proof, we assume that $J$ has a spectrum located in the open right half plane and we assume that $\Lambda$ has a spectrum located in the open left half plane. Using the Floquet factorization, and taking $w(t):=V^{-1}x(t)$, system \eqref{ltih} can be rewritten as:
\begin{align}
	\dot W(t)=(J \otimes \mathcal{I}-\mathcal{N})W(t)+\mathcal{V}^{-1}\mathcal{B}U(t), \label{ltih2}
\end{align} where $W:=\mathcal{F}(w)$.
Thus, the goal now is to solve the Sylvester equation
\begin{align}(J \otimes \mathcal{I}-\mathcal{N})\mathcal{R}-\mathcal{R}(\Lambda \otimes \mathcal{I}-\mathcal{N})&=\mathcal{V}^{-1}\mathcal{B}\mathcal{G}\label{syl1bis}
	\end{align}
Following the assumption concerning the spectra of $(J \otimes \mathcal{I}-\mathcal{N})$ and $(\Lambda \otimes \mathcal{I}-\mathcal{N})$, the solution of \eqref{syl1bis} is provided by (see \cite{Bhatia} for more detail):
$$\mathcal{R}:=\int_0^{+\infty} e^{-(J \otimes \mathcal{I}-\mathcal{N})t}\mathcal{V}^{-1}\mathcal{BG} e^{(\Lambda \otimes \mathcal{I}-\mathcal{N})t}dt.$$
For any $w\in \ell^2$ and taking $\mathcal{G}:=\mathcal{B}^*\mathcal{V}^{*-1}$, we have:
\begin{align*}<w,\mathcal{R}w>&=\int_0^{+\infty} w^*e^{-(J \otimes \mathcal{I}-\mathcal{N})t}\mathcal{G}^*\mathcal{G} e^{(\Lambda \otimes \mathcal{I}-\mathcal{N})t}wdt
\end{align*}
As $\mathcal{N}^*=-\mathcal{N}$, we have:
 $$e^{(\Lambda \otimes \mathcal{I}-\mathcal{N})t}=e^{-(J \otimes \mathcal{I}-\mathcal{N})^*t}e^{((J^*+\Lambda)\otimes \mathcal{I})t},$$ and it follows that:
\begin{align*}<w,\mathcal{R}w>=\int_0^{+\infty}&w^*e^{-(J \otimes \mathcal{I}-\mathcal{N})t}\mathcal{G}^*\mathcal{G}e^{-(J \otimes \mathcal{I}-\mathcal{N})^*t}\\&e^{((J^*+\Lambda) \otimes \mathcal{I})t}wdt\end{align*}
Taking $\Lambda:= -J^*-\alpha Id_n$ with $\alpha\in\mathbb{R}$ such that the spectrum $\sigma(\Lambda)$ is located in the open left half plane, we get:
\begin{align*}<w,\mathcal{R}w>=\int_0^{+\infty}&w^*e^{-(J \otimes \mathcal{I}-\mathcal{N})t}\mathcal{G}^*\mathcal{G}e^{-(J \otimes \mathcal{I}-\mathcal{N})^*t}w\ e^{-\alpha t}dt\end{align*}

As the pair $(\mathcal{A}-\mathcal{N},\mathcal{B})$ is assumed to be exactly controllable, the pair 
$((J \otimes \mathcal{I}-\mathcal{N}),\mathcal{G}^*)$ is also exactly controllable. It follows that for any $t>0$
and any $w\neq0$
$$\int_0^t w^*e^{-(J \otimes \mathcal{I}-\mathcal{N})t}\mathcal{G}^*\mathcal{G}e^{-(J \otimes \mathcal{I}-\mathcal{N})^*t}wdt >0$$

%
%
Consequently, $<w,\mathcal{R}w>$ is strictly postive with $\ w\neq 0$ and hence $\mathcal{R}$ is positive definite and invertible.

To conclude that $\mathcal{P}$, solution of \eqref{syl1}, is invertible, it is sufficient to see that 
 $\mathcal{P}$ satisfies
\begin{align*}\mathcal{P}&:=\int_0^{+\infty} e^{-({A}-\mathcal{N})t}\mathcal{BG} e^{(\Lambda \otimes \mathcal{I}-\mathcal{N})t}dt\\
&:=\int_0^{+\infty} \mathcal{V}e^{-(J \otimes \mathcal{I}-\mathcal{N})t}\mathcal{V}^{-1}\mathcal{BG} e^{(\Lambda \otimes \mathcal{I}-\mathcal{N})t}dt\\
&:=\mathcal{V}\mathcal{R}
\end{align*}
Then, $\mathcal{P}$ is invertible as $\mathcal{V}$ and $\mathcal{R}$ are invertible.
\end{proof}


\subsection{Solving harmonic Sylvester equation}
The harmonic sylvester equation \eqref{syl1} can be solved up to an arbitrarily small error using the same approach as in \cite{Pierre2022} for solving harmonic Lyapunov equations. \color{black}

Let us define the product, denoted by $\circ$, of a $n\times n$ block Toeplitz matrix $\mathcal{A}$ with a matrix $B$ as follows:
\begin{equation}B\circ \mathcal{A}:=\left(\begin{array}{cccc} 
				B\otimes \mathcal{A}_{11} & B\otimes \mathcal{A}_{12} & \cdots & B\otimes \mathcal{A}_{1n} \\
				B\otimes \mathcal{A}_{21} & B\otimes \mathcal{A}_{22} & & \vdots \\
				\vdots & & \ddots& \vdots \\
				B\otimes \mathcal{A}_{n1} & \cdots & \cdots & B\otimes \mathcal{A}_{nn}\end{array}\right)\label{symB}\end{equation}

\begin{theorem}\label{sylvester_approx}
	The phasor ${\bf P}:=\mathcal{F}(P)$ of the solution $\mathcal{P}:=\mathcal{T}(P)$ of the infinite-dimensional Sylvester equation (\ref{syl1}) is given by:
	\begin{equation}col({\bf P}):=(Id_n \otimes (\mathcal{A}-\mathcal{ N})-Id_n\circ (\Lambda\otimes \mathcal{I})^*)^{-1}col({\bf Q})\label{inf_sym}\end{equation}
	with $\mathcal{N}$ given by \eqref{N} and where ${\bf Q}:=\mathcal{F}(BG)$.
\end{theorem}
\begin{proof}The proof is similar to Theorem 10 in \cite{Pierre2022}.
\end{proof}

In practice, due to the infinite dimension nature of our problem, a truncation is necessary. Define for any given $m$:
\begin{enumerate}
\item the $m-$truncation of the $n \times n$ block Toeplitz matrix $\mathcal{A}$ 
obtained by applying a $m-$truncation of all of its blocks $\mathcal{T}(a_{ij})_m$, $i,j:=1,\cdots, n$ where $\mathcal{T}(a_{ij})_m$ is the $(2m+1) \times (2m+1)$ leading principal submatrices of $\mathcal{T}(a_{ij})$.
\item the $m-$truncation ${\bf Q}|_m$ of ${\bf Q}:=\mathcal{F}(Q)$ is obtained by suppressing all phasors (components) of order $|k|>m.$
\end{enumerate}

We define also the $m-$truncated solution as 
\begin{equation}col({\bf \tilde P}_m):=(Id_n \otimes (\mathcal{A}_m-\mathcal{ N}_m)-Id_n\circ (\Lambda\otimes \mathcal{Id}_m)^*)^{-1}col({\bf Q}|_m)\label{inf_sol_2}\end{equation} where $\mathcal{N}_m:=Id_n\otimes diag(j\omega k,\ |k|\leq m)$. 
The next Theorem shows that an approximated solution to the harmonic Sylvester equation can be always determined up to an arbitrarily small error.
\begin{theorem}\label{approx_P_n} Under assumption of Theorem~\ref{sylPP} and assuming that $A(t)\in L^\infty([0 \ T])$, for any given $\epsilon>0$, there exists $m_0$ such that 
	for any $m\geq m_0$: $$\|col({\bf {P}-\tilde {{P}}}_m)\|_{\ell^2}<\epsilon$$
	where ${\bf P}$ is the solution of the infinite-dimensional problem \eqref{inf_sym}.
	Moreover, $$\|\mathcal{P}-\tilde{\mathcal{P}}_m\|_{\ell^2}<\epsilon$$
	with $\mathcal{P}:=\mathcal{T}(P)$ and $\tilde{\mathcal{P}}_m:=\mathcal{T}(\tilde P_m)$.
\end{theorem}
\begin{proof}see Theorem 11 in \cite{Pierre2022}.
\end{proof}

\section{Illustrative example}
\subsection{Periodic trajectory tracking design}
Consider a LTP system defined by
\begin{align}
	\dot x=&\left(\begin{array}{cc}a_{11} (t) & a_{12} (t) \\a_{21} (t) & a_{22} (t)\end{array}\right)x+\left(\begin{array}{c}b_{11}(t) \\0\end{array}\right)u\label{ex_ltp}\end{align}
where {\small\begin{align*}a_{11} (t) &:=1+\frac{4}{\pi}\sum_{k=0}^{\infty}\frac{1}{2k+1}\sin(\omega (2k+1)t),\\
	a_{12} (t) &:= 2+\frac{16}{\pi^2}\sum_{k=0}^{\infty}\frac{1}{(2k+1)^2}\cos(\omega (2k+1)t),\\
	a_{21} (t) &:= -1+\frac{2}{\pi}\sum_{k=1}^{\infty}\frac{(-1)^k}{k}\sin(\omega kt+\frac{\pi}{4}),\\
	a_{22} (t) &:= 1-2\sin(2\pi t)-2\sin(6\pi t)+2\cos(6\pi t)+2\cos(10\pi t),\\
	b_{11}(t)&:=1+ 2 \cos(2\omega t)+ 4 \sin(6\omega t) \text{ with }\omega:=2\pi.
\end{align*}}
Observe that $a_{11}$, $a_{12}$ and $a_{21}$ are respectively square, triangular and sawtooth signals and include an offset part. The associated Toeplitz matrix has an infinite number of phasors and is not banded. 

We first compute the Floquet factorization $(V(t),J)$ (see \cite{Pierre2022}). We find for $J$ a diagonal matrix defined by two complex conjugate eigenvalues $\{1\pm j 1.64\}$ and the component of the $T-$periodic, absolutely continuous and complex valued matrix function $V(t)$ is shown on Fig. \ref{fig5}.
Thus, system \eqref{ex_ltp} can be rewitten as:
$$\dot z=J z+V(t)^{-1}B(t)u$$
and one can see that this LTP system is unstable. 
Moreover, the spectrum of $\mathcal{A}-\mathcal{N}$ is given by $\sigma:=\{\lambda_p+j\omega k: k\in \mathbb{Z}, p:=1,2\}$ with $\lambda_{1,2}:=\{1\pm j 1.64\}$ (see \cite{Pierre2022}).
\begin{figure}[h]\begin{center}
		\includegraphics[width=0.8\linewidth]{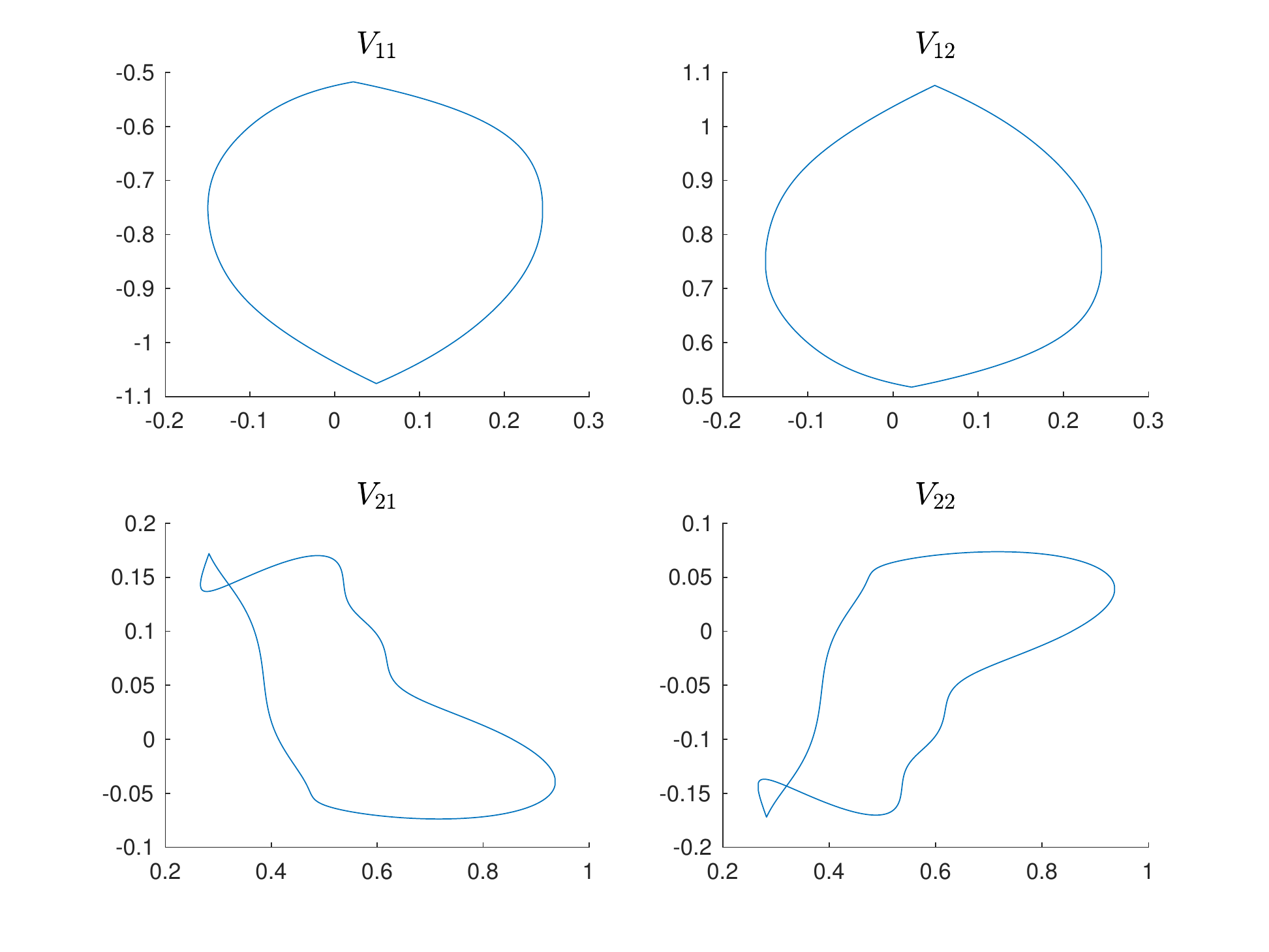}
		\caption{Complex components $V_{ij}(t)$, $i,j=1,2$ of $V(t)$ for $t=[0 \ T]$. }\label{fig5}
	\end{center}
\end{figure}

Using Theorem \ref{suf}, we set $G(t):=B(t)^*V(t)^{*-1}$ and fix the poles of the closed loop by choosing $\Lambda:=-J^*-Id$.
We solve the Harmonic Sylvester equation by computing ${\bf \tilde P}_m$ (see \eqref{inf_sol_2}) for $m:= 4,6,8 ,10$.
Fig.\ref{fig8} shows the components of complex valued $T-$periodic and absolutely continuous matrix function ${\tilde P_m}:=\mathcal{F}^{-1}({\bf \tilde P}_m)$ while Fig. \ref{fig3} shows the phasor modulus of each entries of the matrix ${\bf \tilde P}_m
$ with respect to $m$. 
We see that significant values are obtained for $m\leq 8$.
\begin{figure}[h]\begin{center}
		\includegraphics[width=0.9\linewidth]{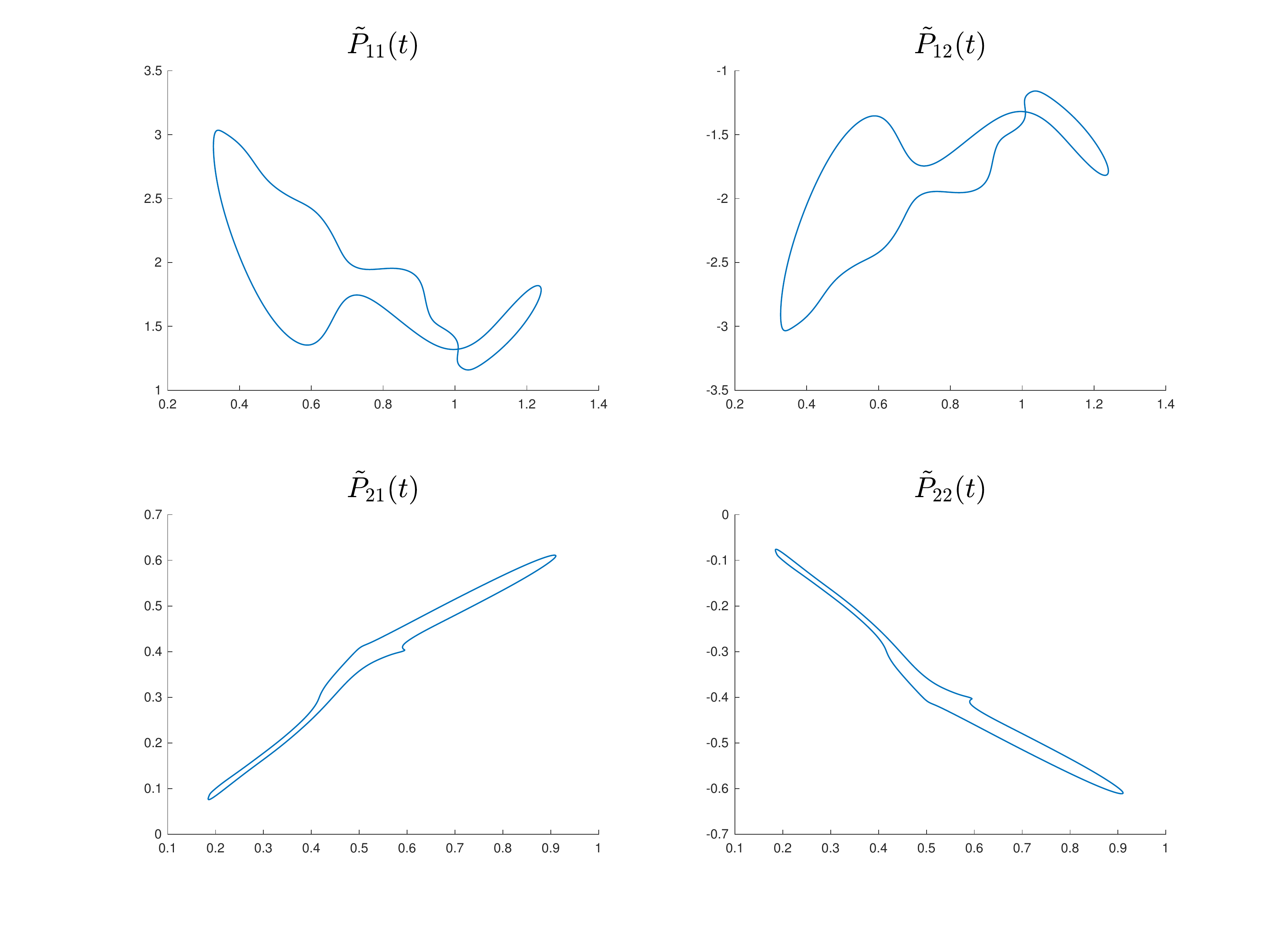}
		\caption{Complex components $\tilde P_{ij}(t)$, $i,j=1,2$ of $\tilde P_m(t)$ for $t=[0 \ T]$ and $m=10$. }\label{fig8}
	\end{center}
\end{figure}
\begin{figure}[h]\begin{center}
		\includegraphics[width=0.9\linewidth]{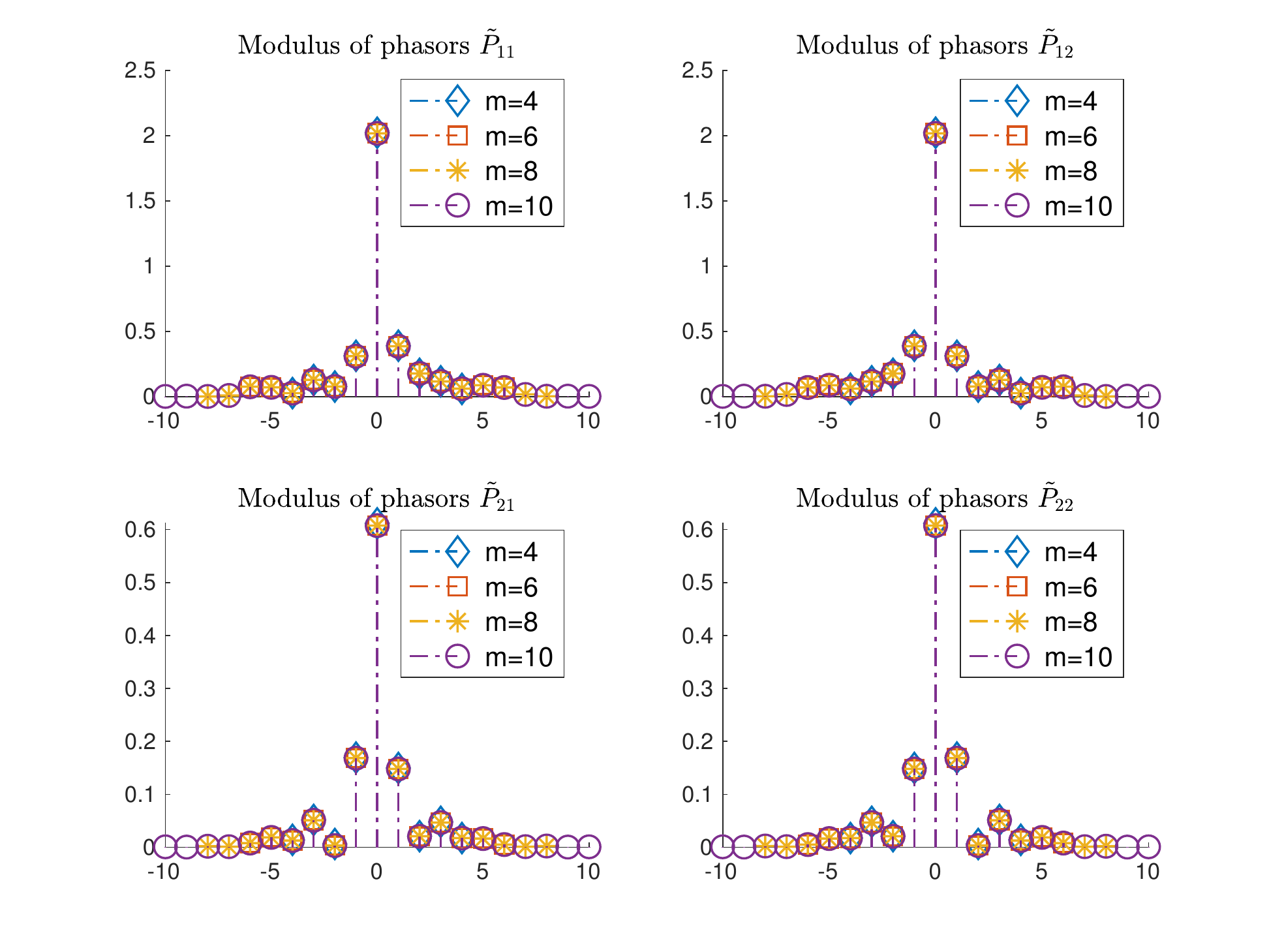}
		\caption{$|{\bf \tilde P}_m|$ for $m:= 4,6,8,10$.}\label{fig3}
	\end{center}
\end{figure}
We deduce the T-periodic and real gain matrix $K(t):=G(t){\tilde P}_m^{-1}(t)$ for $m:= 4,6,8 ,10$ as shown on Fig. \ref{fig7} for $m=10$.
\begin{figure}[h]\begin{center}
		\includegraphics[width=0.8\linewidth,height=6cm]{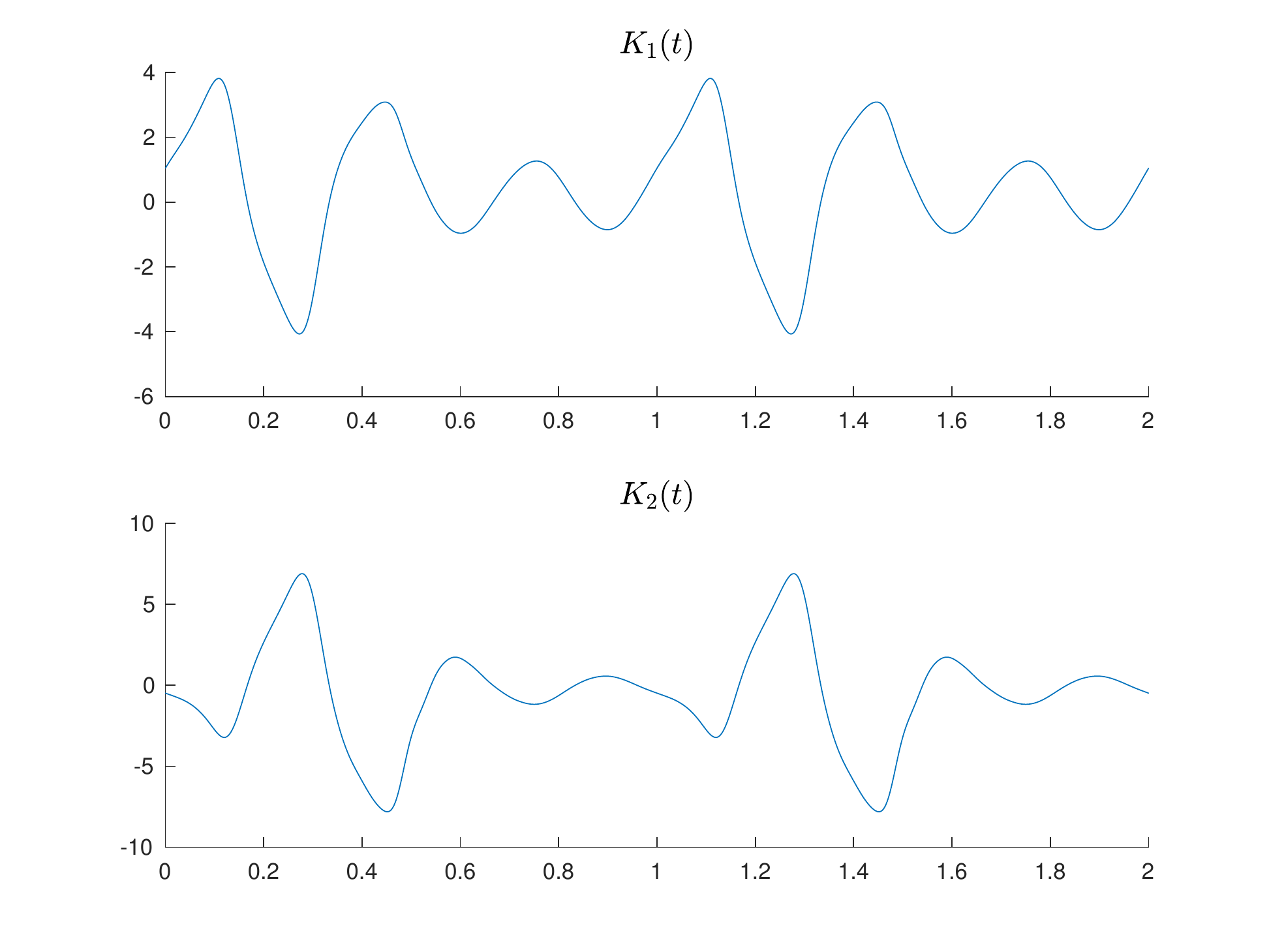}
		\caption{$K(t)=[K_1(t)\ K_2(t)]=G(t){\tilde P}_m^{-1}(t)$ for $m=10$, $t\in[0\ 2T]$.}\label{fig7}
	\end{center}
\end{figure}



%
\begin{figure}[h]\begin{center}
		\includegraphics[width=0.9\linewidth]{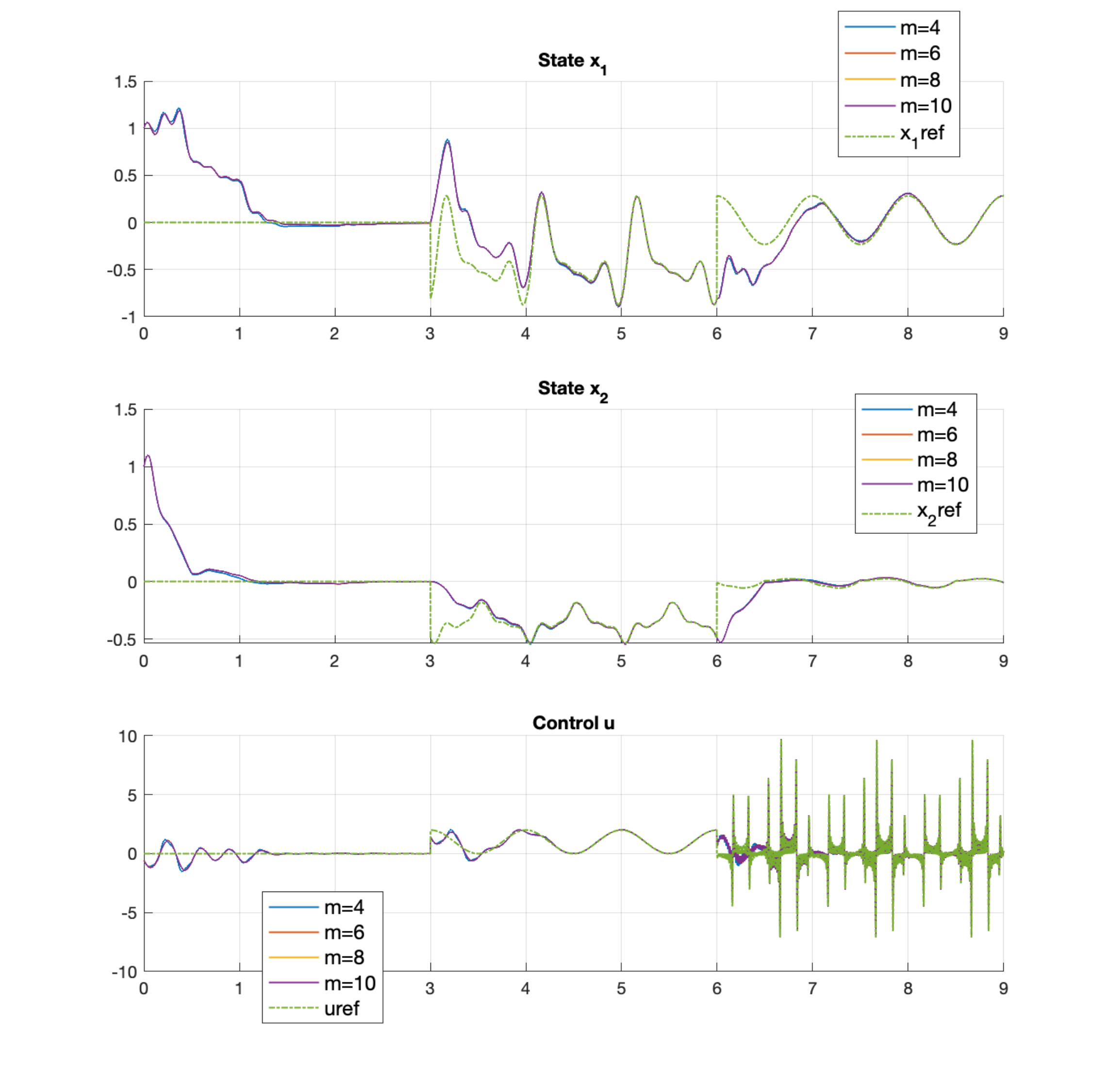}
		\caption{Closed loop response with $u(t):=-K(t)(x(t)-x_{ref}(t))+u_{ref}(t)$ and for $K(t)=\sum_{k=-{m}}^{m} K_k e^{j\omega kt}$ and $m:=4,6,8,10$ and $\Lambda:=-J^*-Id$ and $ J:=diag(1\pm j 1.64)$.}\label{fig2}
	\end{center}
\end{figure}


Now, it is straightforward to show that the control $$u(t):=-K(t)(x(t)-x_{ref}(t))+u_{ref}(t),$$ stabilizes globally and asymptotically the unstable LTP system \eqref{ex_ltp} on any $T-$periodic trajectory $x_{ref}(t):=\mathcal{F}^{-1}(X_{ref})$ and $u_{ref}(t):=\mathcal{F}^{-1}(U_{ref})$ where the pair $(X_{ref},U_{ref})$ satisfies the harmonic equilibrium 
equation
\begin{align}0=(\mathcal{A}-\mathcal{N})X_{ref}+\mathcal{B}U_{ref}.\label{equi}\end{align}
To illustrate this, we plot on Fig.~\ref{fig2} the closed loop response for three $T-$periodic reference trajectories $(x_{ref},u_{ref})$. We start by $u_{ref}(t):=0$, for $t<3$, then $u_{ref}:=1+\cos(2\pi t)$ for $3\leq t< 6$ and for $t\geq 6$ we consider a desired steady state $X_d$ given by $\mathcal{F}^{-1}(X_d)(t):=(\frac{1}{4}\cos(2\pi t),0)$ and look for the nearest harmonic equilibrium, solution of the minimization problem $\min_{U_{ref}}\|X_d-X_{ref}\|^2$ subject to \eqref{equi}. Clearly it can be observed on Fig.~\ref{fig2} that the provided state feedback allows to track any $T-$periodic trajectory corresponding to any equilibrium of \eqref{equi} even if a relative small number $m$ of involved harmonic are considered. 

Now, it can be noticing that Theorem~\ref{suf} provides only a sufficient condition for the invertibility of $\mathcal{P}$. In fact, this invertibility can be obtained for many other pole placements. For example, if we set $\Lambda :=diag(-10,-12)$ instead of $\Lambda:=-J^*-\alpha Id$, from $\tilde P_m(t):=\mathcal{F}^{-1}({\bf \tilde P}_m)=\sum_{k=-{m}}^{m} {\bf\tilde P}_k e^{j\omega kt}$ which is a $T-$periodic and continuous function, it is easy to check that 
$|\det {\tilde P}_m(t)|>0$ for any $t\in[0\ 1]$ $(T=1)$ and thus $\tilde P_m(t)$ is invertible. The corresponding results depicted in Fig.\ref{fig6} illustrate the improvement of the transient.
\begin{figure}[h]\begin{center}
		\includegraphics[width=0.9\linewidth]{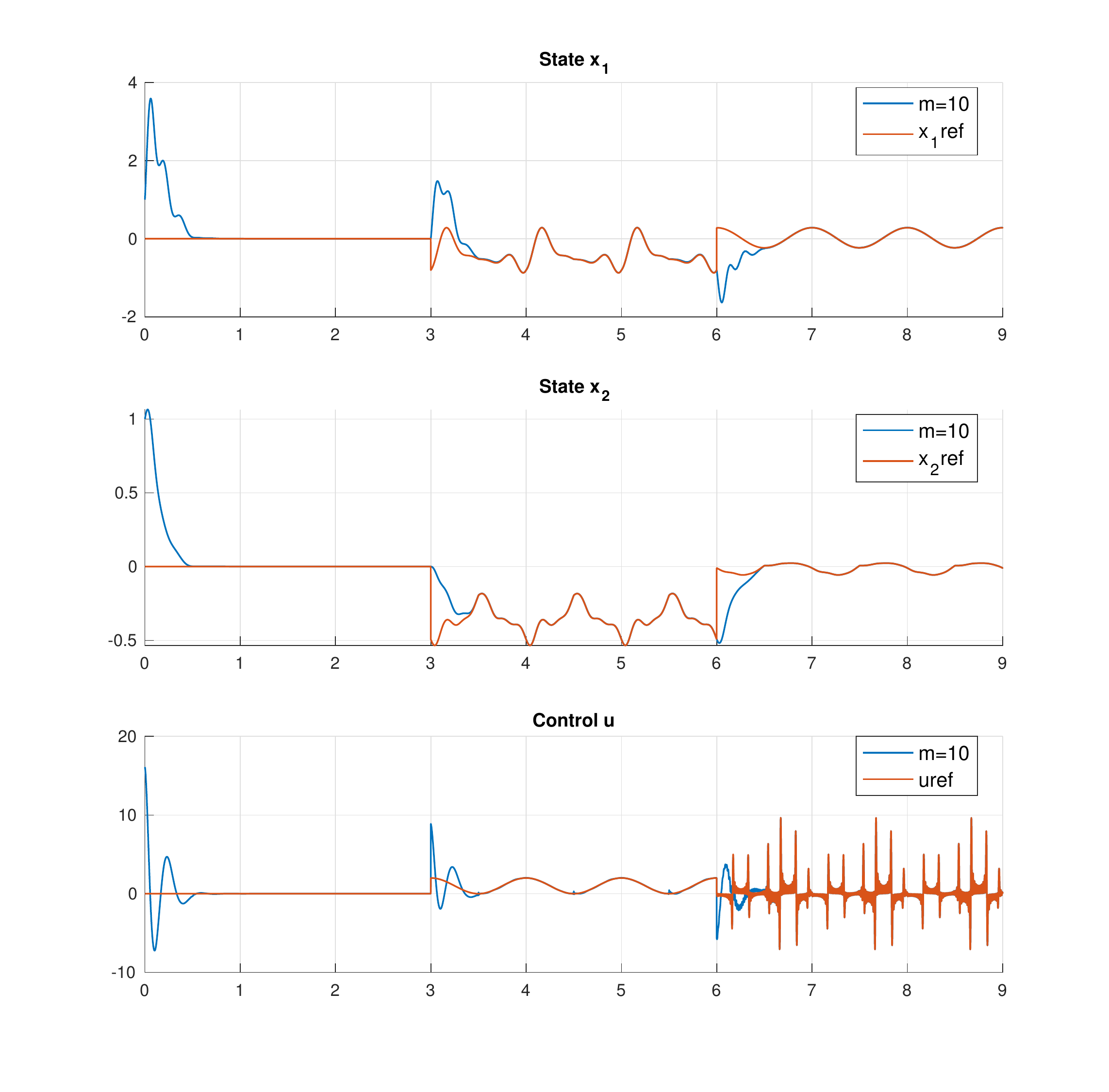}
		\caption{Closed loop response with $u(t):=-K(t)(x(t)-x_{ref}(t))+u_{ref}(t)$ for $K(t)=\sum_{k=-{m}}^{m} K_k e^{j\omega kt}$, $m:=10$ and $\Lambda:=diag(-10,-12)$.}\label{fig6}
	\end{center}
\end{figure}

\subsection{Open question}
Now, we return to the question related to the invertibility of $\mathcal{P}$ as discussed in Remak~\ref{inv} and in Theorem~\ref{suf}. A counter-example is now given to show that the observability condition used in the finite dimension case is not sufficient to enforce invertibility of $\mathcal{P}$.
Let $G(t):= [1\ 1]$ and fix the poles of the closed loop to $\Lambda:=diag(-5,-7)$. Then, it is obvious that the constant pair $(G,\Lambda)$ is observable.
Moreover, as $\mathcal{G}:=[\mathcal{I}\ \mathcal{I}]$, the pair $(\mathcal{G},(\Lambda \otimes \mathcal{I}-\mathcal{N}))$ is exactly observable since it satisfies
 at time $t$, $$\exists \delta>0:\forall x\in \ell^2,\int_0^t\|\mathcal{G}e^{(\Lambda \otimes \mathcal{I}-\mathcal{N}))\tau}x\|_{\ell^2}^2 d\tau \geq \delta \| x \|_{\ell^2}^2.$$
Indeed, as the $k-$th component of $\mathcal{G}e^{(\Lambda \otimes \mathcal{I}-\mathcal{N}))\tau}x$ is given by $x_{1,k}e^{(\lambda_1+j\omega k)t}+x_{2,k}e^{(\lambda_2+j\omega k)t}$, the only solution $(x_{1,k},x_{2,k})$ that provides $x_{1,k}e^{(\lambda_1+j\omega k)t}+x_{2,k}e^{(\lambda_2+j\omega k)t}=0$ on any time intervall $[0\ t]$ is $(x_{1,k},x_{2,k}):=(0,0)$ since $\lambda_1\neq \lambda_2$.

Now, if we compute $\mathcal{P}$ using $\eqref{inf_sol_2}$, the determinant $\det(P(t))$ of the continuous function $P(t)$ vanishes for many $t\in[0\ 1]$ and Theorem~\ref{inverse} implies that $P(t)$ and $\mathcal{P}$ are not invertible. In fact, $P(t)$ is almost everywhere invertible as shown in Fig.~\ref{fig4}.
The non invertibility of $P$ can be explained as follows.
If we replace $K(t)$ by $G(t)P^{-1}(t)$ in the differential Sylvester equation
$$\dot P^{-1}(t)=-P^{-1}(t)(A(t)-B(t)K(t))+\Lambda P^{-1}(t)\ a.e.$$
we obtain the following quadratic differential equation
\begin{equation}\dot P^{-1}(t)=-P^{-1}(t)A(t)+\Lambda P^{-1}(t)+P^{-1}(t)B(t)G(t)P^{-1}(t)\label{ric}\end{equation}
and integrating this non linear equation leads to a finite escape time. As a consequence, $P^{-1}(t)$ and thus $K(t)$ are not $L^\infty$.
This example leaves open the question related to how to choose $G(t)$ in order to guarantee the existence of a $T-$periodic solution of \eqref{ric}.

\begin{figure}[h]\begin{center}
		\includegraphics[width=0.9\linewidth,height=5cm]{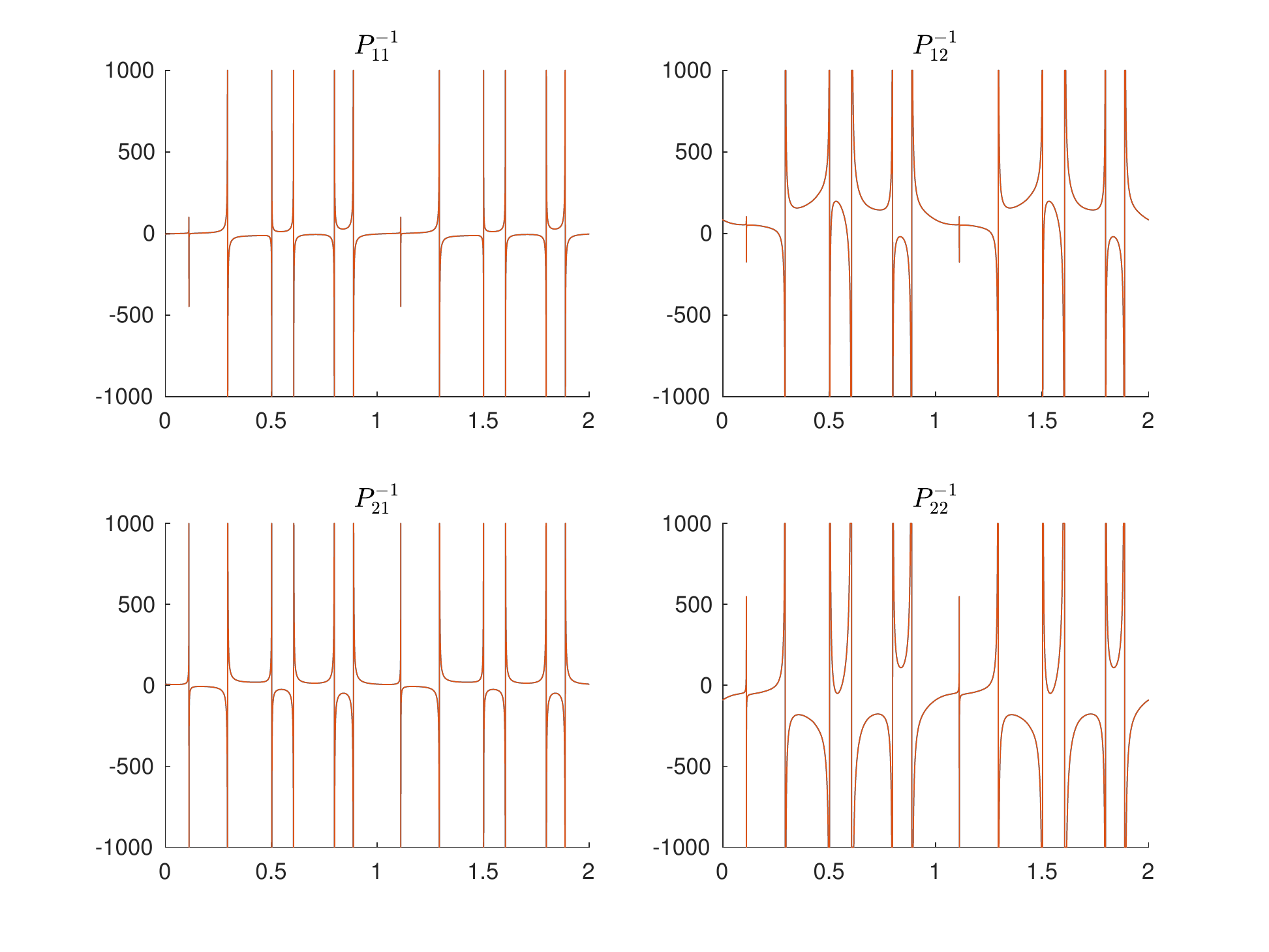}
		\caption{Plot of $P^{-1}(t)$: Counter example where $P(t)$ is almost everywhere invertible and where $\mathcal{P}$ is not invertible despite that the observability condition is satisfied.}\label{fig4}
	\end{center}
\end{figure}

\section{Conclusion}
A harmonic pole placement procedure has been proposed. It allows to design control laws with performance features to stabilize LTP systems on any periodic trajectory satisfying the harmonic equilibrium equation. Provided that the solution of the Sylvester harmonic equation is invertible, the resulting closed-loop system has an LTI representation with a specified pole placement. 
A sufficient condition is given to ensure this invertibility. From a practical point of view, the infinite-dimensional Sylvester equation is solved up to an arbitrarily small error.
Finally, an example is given to show both the features of the proposed procedure and the fact that unlike the finite dimensional case, an observability condition is not sufficient to ensure the invertibility of the solution of the Sylvester equation.

\end{document}